\documentclass[letterpaper, 10 pt, conference]{ieeeconf} 
\IEEEoverridecommandlockouts
\usepackage[neverdecrease]{paralist}
\usepackage{amsmath}
\usepackage{color}
\usepackage{algorithm} 
\usepackage{algpseudocode}
\usepackage{tikz}
\usepackage{graphicx}
\usepackage{authblk}
\usepackage{mathtools}
\usepackage{amsfonts}
\usepackage{caption}
\usepackage[noadjust]{cite}

\newtheorem{lemma}{Lemma}
\newtheorem{remark}{Remark}
\newtheorem{assum}{Assumption}
\newtheorem{theorem}{Theorem}

\newtheorem{prop}{Proposition}

\DeclareGraphicsExtensions{.png,.jpg,.eps}

\title{\LARGE \bf  Conic-sector-based analysis and control synthesis \\ for linear parameter varying systems}

\author{Sivaranjani S, James Richard Forbes, Peter Seiler and Vijay Gupta
\thanks{Sivaranjani S and Vijay Gupta are with the Department of Electrical Engineering, University of Notre Dame, South Bend, IN. \{sseethar,vgupta2\}@nd.edu. James Richard Forbes is with the Department of Mechanical Engineering, McGill University, Montreal, QC, Canada. james.richard.forbes@mcgill.ca. Peter Seiler is with the Department of Aerospace Engineering and Mechanics, University of Minnesota, Minneapolis, MN. seile017@umn.edu. 
}
}

\begin{document}

\maketitle

\thispagestyle{empty}
\pagestyle{empty}

\begin{abstract}
We present a conic sector theorem for linear parameter varying (LPV) systems in which the traditional definition of conicity is violated for certain values of the parameter. We show that such LPV systems can be defined to be conic in an average sense if the parameter trajectories are restricted so that the system operates with such values of the parameter sufficiently rarely. We then show that such an average definition of conicity is useful in analyzing the stability of the system when it is connected in feedback with a conic system with appropriate conic properties. This can be regarded as an extension of the classical conic sector theorem. Based on this modified conic sector theorem, we design conic controllers that allow the closed-loop system to operate in nonconic parameter regions for brief periods of time. Due to this extra degree of freedom, these controllers lead to less conservative performance than traditional designs, in which the controller parameters are chosen based on the largest cone that the plant dynamics are contained in. We demonstrate the effectiveness of the proposed design in stabilizing a power grid with very high penetration of renewable energy while minimizing power transmission losses.
\end{abstract}

\section{Introduction}
The conic sector theorem, introduced by Zames \cite{zamescombined}, is a powerful input-output stability result that is applicable to both linear and nonlinear systems that are sector bounded. Traditionally, conic-sector-based analysis has received relatively limited attention in comparison to small gain and passivity-based analysis due to the difficulty involved in accurately characterizing conic bounds.  However, recent results in the efficient computation of conic bounds have made it possible to consider the conic sector theorem as an attractive tool for analysis and control design \cite{bridgeman2014conic}.% for control design \cite{bridgeman2014conic}. 
%While analysis and control synthesis approaches based on the small gain and passivity theorems have been widely studied, conic-sector-based analysis has traditionally received relatively less attention due to the difficulty involved in accurately characterizing the conic bounds for a system. However, recent results in the efficient computation of conic bounds have now made it possible to consider the conic sector theorem as an attractive tool for control design (see, for example, \cite{bridgeman2014conic}). 
%For example, Bridgeman et al. \citep{bridgeman2014conic} design conic controllers to reestablish closed loop input-output stability for a class of systems that are nonpassive in the open loop.  

In this paper, we present a conic-sector-based stability analysis and control synthesis for linear parameter varying (LPV) systems. %Linear parameter varying systems are an important modeling paradigm, finding applications in aerospace engineering, power grids and various other fields \cite{leith1998gain,lawrence95,shamma90}. 
Typical control design approaches for LPV systems involve designing and scheduling a bank of controllers to stabilize the system for all parameter values in its operating trajectory \cite{leith1998gain,lawrence95,shamma90}. Analysis and synthesis results for LPV systems based on small gain theorems \cite{apkarian1995,packard94,scherer2012a} and dissipativity theory \cite{apkarian98,packard92,wu1995} are also available. Of particular interest to this paper, controllers have been designed for LPV systems \cite{joshi2002design} and specific subclasses like polytopic systems \cite{alexthesis} using the conic sector theorem. However, these controllers are usually conservative since they are designed for the worst-case sector bounds on the plant as a function of the parameter. Furthermore, these results cannot be used for LPV systems in which the system is non-conic for even one value of the parameter. In this paper, we propose a less conservative design approach based on a modification of the conic sector theorem that does not constrain the designer to consider the worst-case sector bound.
%controllers usually lead to a conservative design, as they are designed for the worst-case sector bounds on the plant as a function of the plant parameter. 

Recent results indicate that passivity in LPV systems can be preserved on the average even when certain parameter values lead to nonpassive operation, provided that the system spends a sufficiently small amount of time operating with these parameter values \cite{sivaranjanipassivity}. %These results are also similar to the stability analysis of switched systems in which some modes are unstable \cite{wyb02,hespanha2004linear}. 
Here, we generalize these results to obtain a modified conic sector theorem for intermittently conic LPV systems, that is, LPV systems that may be nonconic for certain parameter values. We then use this theorem to design a conic controller that allows the closed-loop LPV system to be intermittently conic. This approach allows for less conservative control designs as the controller is not designed for the worst case sector bound of the plant, which, in fact, does not exist for intermittently conic plants.

The contributions of this paper are threefold. Firstly, we derive conditions under which an LPV system with intermittent conic behavior is conic in the average sense. Secondly, we derive a modified conic sector theorem for the feedback interconnection of an intermittently conic LPV system and a conic system. Finally, we design a conic controller based on this conic sector theorem that allows the closed-loop system to be intermittently conic, leading to less conservative designs than methods available in literature. We also demonstrate an application of the proposed design in stabilizing a power system with high penetration levels of renewable energy while minimizing power losses.%and illustrate its performance and advantages by simulation on a modified IEEE 14-bus test system.

%This paper is organized as follows. Section \ref{problem}, we present the system model and formally characterize conicity for LPV systems. In Section \ref{exact}, we derive conditions for system conicity in terms of the fractions of times spent in the conic and nonconic parameter regions. In Section \ref{conicsector}, we derive a modified conic sector theorem. In Section \ref{controldesign}, the design of conic controllers based on the modified conic sector theorem to stabilize an LPV system and satisfy certain desired performance objectives is presented. In Section \ref{simulation}, we illustrate the conic controller design using the example of the IEEE 14-bus test power system with integrated renewable energy.

\textit{Notation:} $\mathbb{R}$ denotes the set of real numbers, $\mathbb{R}^{p}$ the set of $p$-dimensional real vectors, $\mathbb{R}^{+}$ the set of non-negative real numbers, $\mathbb{N}$ the set of natural numbers including zero and $\mathbb{N}_n$ the set of natural numbers $\{1,\cdots,n\}$, $n \in \mathbb{N}$. Given two sets $A$ and $B$, $A\setminus B$ denotes the set of all elements of $A$ that are not in $B$. For a vector $v$, $[v]_{i}$ denotes the $i$-th component. $I$ stands for the identity matrix with dimensions clear from the context. $A'$ denotes the transpose of a matrix $A$.

\section{Problem Formulation}\label{problem}
\textbf{Linear Parameter Varying (LPV) system:} Consider the system described by
\vspace{-0.5em}
\begin{equation}
\begin{gathered} \label{lpvsys}
\dot{x}(t) = A({\rho}(t))  x(t) + B({\rho}(t))  u(t) \\
y(t) = C(\rho(t))  x(t) + D(\rho(t))  u(t)
\end{gathered}
\vspace{-0.5em}
\end{equation}
where $x(t)\in\mathbb{R}^{n}$ is the process state, $u(t)\in\mathbb{R}^{m}$ is the control input, $y(t)\in\mathbb{R}^{m}$ is the process output, and $\rho\in \mathbb{R}^{p}$ is the parameter vector which is assumed to be a piecewise continuously differentiable function of time.  We will limit consideration to the time interval $[t_0,t_n]$ and consider $t_0=0$ and $t_n\rightarrow\infty$ to define stability for this system. 
The parameter trajectory from times $t_0$ to $t_n$ (denoted by $\rho(t)|_{t_0}^{t_n}$ or simply $\rho(t)$  when the interval $[t_0,t_n]$ is clear from context) refers to the set of parameter values that the system assumes from time $t_0$ to $t_n$. The following assumptions will hold throughout the paper.
\begin{assum}\label{input}
	We assume that the input $u(t)$ is norm bounded so that $||u(t)||_2^2 \in [\bar u_{1}^2, \bar u_{2}^2]$ where $\bar u_{1}>0$ and $\bar u_{2}<\infty$ are positive constants.
\end{assum}

LPV models are often constructed as linearizations of nonlinear dynamics about a collection of different operating points. Assumption \ref{input} implies that $u(t)=0$ for all $t$ is not allowed, that is, we disallow operation at an equilibrium for such LPV models.

\begin{assum}
	The parameter trajectories satisfy the following range and rate bounds
	\vspace{-0.5em}
	\begin{equation}\label{rhobound}
%	\begin{gathered}
	\underline{\rho} \le [\rho(t)]_{i} \le \bar{\rho}, \quad
	\underline{\nu}  \le \left[\frac{d \rho(t)}{dt}\right]_{i} \le \bar{\nu},\quad \forall 1\leq i\leq p.
%	\end{gathered}
	\end{equation}
\end{assum}
%with the range and rate bounds defining hyperrectangles in $\mathbb{R}^{p}$ denoted by $\mathcal{P}_p$ and $\mathcal{\dot{P}}_p$ respectively. 
The special case when $\underline{\nu}\rightarrow-\infty$ and $\bar{\nu}\rightarrow\infty$ is called the rate unbounded case. 

Let all parameter trajectories that satisfy these constraints be collected in the set $\mathcal{A}$, which is the set of \textit{admissible parameter trajectories}. 
%The parameter trajectory is said to be rate
%unbounded if ${\dot{\mathcal{P}}_p} = \mathbb{R}^{n_{\rho}}$. Let 
%$\mathcal{A}_p := \{  \rho_p(t): \mathbb{R}^+ \rightarrow \mathbb{R}^{p}
%\ : \ \rho_p(t) \in \mathcal{P}_p, \ \dot{\rho}_p(t) \in \mathcal{\dot{P}}_p \
%\forall t \ge 0 \}$ be the set of admissible parameter trajectories. 
An admissible parameter trajectory $\rho(t)|_{t_0}^{t_n} \in \mathcal{A}$ where the parameter assumes constant values $\rho(i)$ indexed by $i \in \{1,2,\ldots,n\}$, $n \in \mathbb{N}$, for time $t_{i-1}< t \leq t_i$ is defined as a discrete parameter trajectory over the time interval $[t_0,t_n]$. We denote the set of all \textit{discrete parameter trajectories} by $\mathcal{A}^d$, where $\mathcal{A}^d \subset \mathcal{A}$.

\textbf{Conicity:} The system~(\ref{lpvsys}) with a fixed value of the parameter is a linear time-invariant system and can be analyzed for conicity~\cite{bridgeman2014conic}. As the parameter value varies, we use the following definition of conicity %the system becomes time-varying. We thus use the following definition of conicity %Existing results{~\citep{joshi2002design} have considered the case when for any fixed parameter value, the resulting system is conic. We allow the LPV system to be nonconic for some parameter values. We therefore define conicity in the average sense 
for particular parameter trajectories.
%\begin{defn}%[Conicity in the Average Sense]
%\label{conicity}
%Specifically, 
The LPV system (\ref{lpvsys}) with a given parameter trajectory $\rho(t)|_{t_0}^{t_n}$ is said to be conic for that parameter trajectory if there exist constants $a$ and $b$ such that $0<a<b$ and the inequality 
\vspace{-1.2em}
\begin{equation}
\label{eq:conicity}
%\frac{1}{t_N-t_0}
\mathop \int_{t_1}^{t_2} w(u(t),y(t))dt \geq 0,
\vspace{-0.7em}
\end{equation}
holds for all times~$t_1$ and~$t_2$ such that~$t_0\leq t_1 <t_2\leq t_n$ and all $u(t)$ satisfying Assumption \ref{input}, where \setlength{\arraycolsep}{-0.8pt}
\vspace{-0.7em}
\begin{equation*}
\label{eq:w}
%\frac{1}{t_N-t_0}
w(u(t),y(t))~=~\left[\begin{array}{c} y(t) \\ u(t)\end{array}\right]' 
\left[ \begin{array}{lr}-\frac{1}{b}I&\frac{1}{2}\left(1+\frac{a}{b}\right)I \\\frac{1}{2}\left(1+\frac{a}{b}\right)I&-aI\end{array}\right]
\left[\begin{array}{c}y(t) \\u(t)\end{array}\right].
\end{equation*}
Further, if the system is conic for every parameter trajectory in a set $\mathcal{S} \subseteq \mathcal{A}$, then it is said to be conic for that set $\mathcal{S}$. Note that the constants $a$ and $b$ will, in general, depend on the parameter trajectory and are not unique. We assume that the conic bounds represent the tightest bounds for the set of parameter trajectories $\mathcal{S}\subseteq \mathcal{A}$ and time interval being considered, and can be characterized by solving the matrix inequality (\ref{coniclmi}) as described in Appendix \ref{determining}. We define the interval $[a,b]$ as the conic sector for the system (\ref{lpvsys}) for the set $\mathcal{S}$, the value $(b-a)/2$ as the corresponding radius $r$ of the conic sector, and the quantity $(b+a)/2$ as the center $c$ of the conic sector. We also drop the dependence of $w(u(t),y(t))$ on time and use $w(u,y)$ for brevity. Note that \eqref{eq:conicity} is a finite horizon integral quadratic constraint (IQC) on $(u,y)$. The most general IQC theory also allows dynamic mulitpliers, with \eqref{eq:conicity} being the special case where the multiplier is simply the (static) identity matrix \cite{megretski1997system}. 

\textbf{Intermittent conicity:} Equation~\eqref{eq:conicity} constrains the parameter trajectory such that~(\ref{eq:conicity}) holds for all choices of times $t_{1}$ and $t_{2}$ where $t_0\leq t_1 <t_2\leq t_n$. We are interested in trajectories that are only `intermittently conic' in the following sense. For a given parameter trajectory  $\rho(t)|_{t_0}^{t_n},$ identify times $t_0\leq t_{1}<t_{2} <\cdots \leq t_{n}$ such that for every $i\in i_c \subset \mathbb{N}_n$, the system is conic for the parameter trajectory $\rho(t)|_{{t}_{i-1}}^{{t}_{i}}$, and for every $i \in i_{nc}$, where $i_{nc} =\mathbb{N}_n \setminus i_c$, the system is not conic for the parameter trajectory $\rho(t)|_{{t}_{i-1}}^{{t}_{i}}$. Note that the choice of these times is not unique. The \textit{conic parameter region} of the trajectory $\rho(t)|_{t_0}^{t_n}$ is then defined as the set $R_c=\{\rho(t):t\in\bigcup \limits_{i \in i_c}\left[t_{i-1},t_{i}\right]\}$ and the time spent in this region is denoted by  $t_c$. The \textit{nonconic parameter region} of the trajectory $\rho(t)|_{t_0}^{t_n}$ is the set $R_{nc}=\{\rho(t):t\in[t_0,t_n],t\notin R_{c}\}$ and the time spent in this region is denoted by $t_{nc}$. If the set $R_{c}$ is not empty, the system is \textit{intermittently conic} for the parameter trajectory $\rho(t)|_{t_0}^{t_n}$. 

For intermittently conic systems, we extend the definition of conicity proposed in \eqref{eq:conicity} as follows. The LPV system (\ref{lpvsys}) with a given parameter trajectory $\rho(t)|_{t_0}^{t_n}$ is said to be \textit{conic in the average sense} for that parameter trajectory if there exist constants $a$ and $b$ with $0<a<b$ such that
\vspace{-0.5em}
\begin{equation}
\label{eq:av_conicity}
%\frac{1}{t_N-t_0}
\mathop \int_{t_0}^{t_n} w(u,y)dt \geq 0
\vspace{-0.5em}
\end{equation}
holds for all $u(t)$ satisfying Assumption \ref{input}. If the system is conic in the average sense for every parameter trajectory in a set $\mathcal{S} \subseteq \mathcal{A}$, then it is said to be conic in the average sense for that set $\mathcal{S}$. We will make the following assumptions on the conic and nonconic parameter regions.% considered in this paper. 
\begin{assum}
	\label{ass:conic}
	For every $i\in i_c$, there exist constants $a$ and $b$ and a constant $\epsilon_{i}$ such that $0<a<b$ and
	\vspace{-0.7em}
	\begin{equation}
	\label{conic_eqn}
	\mathop \int_{t_{i-1}}^{t_{i}} w(u,y)dt \geq \epsilon_{i}(t_{i}-t_{i-1}).
	\end{equation}
\end{assum}
\begin{remark}
{	Assumption \ref{ass:conic} is similar to the notion of quasi-QSR dissipativity with the supply function given by the left hand side of~(\ref{conic_eqn}) (see, e.g.,~\cite[Chapter 4]{lozano2013dissipative}). However, by Assumption \ref{input}, we do not allow uniformly zero inputs over any finite time horizon. Therefore, classical constraints on the sign of $\epsilon_i$ (\cite[Chapter 4]{lozano2013dissipative}%,\cite[Chapter 7]{haddad2008nonlinear}
	) are not applicable in our case. We note that Assumption \ref{ass:conic} may be relaxed to allow zero inputs over specific time intervals by imposing the sign constraint $\epsilon_{i}<0$ in such intervals. In this paper, we avoid such a formulation due to increased conservatism. %However, such a formulation is avoided since it introduces conservatism in the results by mis-classifying some conic parameter regions as being nonconic.
} 
\end{remark}
%Since conicity implies quasi-QSR dissipativity with the supply function given by the left hand side of~(\ref{conic_eqn}) (see, e.g.,~\cite[Chapter 4]{lozano2013dissipative}), Assumption~\ref{ass:conic} can be interpreted as requiring the system to strictly dissipate energy whenever it is in the conic region of the parameter trajectory. 
\begin{assum}\label{nonconic_assumption}
	For every $i\in i_{nc}$, there exists %constants $a$ and $b$ and 
	a constant $\alpha_{i}>0$ such that %$0<a<b$ and
		\vspace{-0.68em}
		\begin{equation}
		\label{nonconic_eqn}
	\mathop \int_{t_{i-1}}^{t_{i}} w(u,y)dt\geq -\alpha_{i} \left(t_{i}-t_{i-1}\right) \int_{t_{i-1}}^{t_{i}} u^2(t)dt .
	\end{equation}
\end{assum} 
Note that the $\int_{t_{i-1}}^{t_{i}} u^2(t)dt$ term in \eqref{nonconic_eqn} ensures that no conic parameter region satisfies \eqref{nonconic_eqn}. In the context of Assumption \ref{nonconic_assumption}, $a$ and $b$ represent the smallest constants such that \eqref{nonconic_eqn} holds with $0<a<b$, rather than conventional sector bounds. The constants $\epsilon_{i}$ and $\alpha_{i}$ depend on the parameter trajectory $\rho(t)|_{t_{i-1}}^{t_{i}}$; we drop this dependence for simplicity of notation. {In general, we jointly compute the tightest bounds for $\epsilon_{i}$ and $\alpha_{i}$ with the sector bounds $[a,b]$ (see Appendix \ref{determining}).}

\textbf{Problem statement:} With these assumptions,  the aim of this paper is to
\begin{itemize}
	\item[(i)] derive conditions under which an LPV system (\ref{lpvsys}) with intermittently conic behavior for a parameter trajectory $\rho(t)|_{t_0}^{t_n}$ is conic in the average sense for that trajectory, 
	\item[(ii)] derive a conic-sector theorem that ensures $\mathcal{L}_2$ stability when system (\ref{lpvsys}) with intermittent conic behavior for a parameter trajectory is connected in negative feedback with another conic system, and
	\item[(iii)] design a conic controller using this conic sector theorem to stabilize an intermittently conic LPV system while guaranteeing some gain performance.
\end{itemize}
\section{Conicity of LPV Systems with Intermittent Conic Behavior}\label{exact}
We begin by deriving conditions for average system conicity in terms of the fractions of time spent in conic and nonconic parameter regions. 
\begin{lemma}[Discrete Case]\label{res1}
	%\textit{
	Let the parameter trajectories of the intermittently conic system (\ref{lpvsys}) be restricted to the set of discrete parameter trajectories $\mathcal{A}^d \subset \mathcal{A}$. Then, \eqref{lpvsys} is conic in the average sense if, for every $\rho(i)|_{t_0}^{t_n} \in \mathcal{A}_d$, we have
	\setlength{\arraycolsep}{0pt}
	\vspace{-0.5em}
	{
		\begin{equation}\label{res_1}
		\begin{gathered}
		\sum \limits_{i=1}^{n} c_i\mu_i \geq 0, \, \mu_i=\frac{t_{i}-t_{i-1}}{t_n-t_0}, \, {c_i} = \left\{ \begin{array}{*{20}{c}}
		{\epsilon_{i},\,i\in i_{c}} \\ 
		{-\alpha_{i}\bar{u}_1^2,\,i \in i_{nc} }
		\end{array} \right. .
		\end{gathered}
		\end{equation}
	}
\end{lemma}
We can readily generalize Lemma \ref{res1} to the case where the LPV system \eqref{lpvsys} can operate along any admissible parameter trajectory as follows.
\begin{theorem}\label{res2}
	The system (\ref{lpvsys}) is conic in the average sense if, for every $\rho(t)|_{t_0}^{t_n} \in \mathcal{A}$, we have
%	\vspace{-0.5em}
	\begin{equation}
	\begin{gathered}\label{res_p}
	\int \limits_{t \in t_c} \epsilon(\rho(t)) dt \geq
	\bar{u}_1^2\int \limits_{t \in t_{nc}}\alpha(\rho(t)) dt,
	\end{gathered}
	\end{equation}
	%	}
	where $\epsilon(\rho(t))$ and $\alpha(\rho(t))$ are continuous approximations of $\epsilon_i$ and $\alpha_i$ as defined in Appendix \ref{proofs}.
\end{theorem}

The proofs of Lemma \ref{res1} and Theorem \ref{res2}, along with the formal definitions of $\epsilon(\rho(t))$ and $\alpha(\rho(t))$, can be found in Appendix \ref{proofs}. Lemma \ref{res1} and Theorem \ref{res2} formalize the idea that the LPV system (\ref{lpvsys}) is conic in the average sense if it spends a sufficiently large fraction of time operating in conic parameter regions.

\section{Conic Sector Theorem for Intermittently Conic Systems}\label{conicsector}
Consider a negative feedback interconnection of LPV systems $G_1$ and $G_2$ as shown in Fig. \ref{negativefeedback}, where the dynamics of $G_1$ and $G_2$ are described as follows.
%\vspace{-0.5em}
\begin{align}\label{lpvsys1}
G_1: \quad  &\dot{x}_p(t) = A_p({\rho}_p(t))  x_p(t) + B_p({\rho}_p(t))  u_p(t) \nonumber\\
&y_p(t) = C_p(\rho_p(t))  x_p(t) + D_p(\rho_p(t))  u_p(t)
\end{align}
\vspace{-2.2em}
\begin{align}\label{lpvsys2}
G_2: \quad  &\dot{x}_c(t) = A_c({\rho}_c(t))  x_c(t) + B_c({\rho}_c(t))  u_c(t) \nonumber\\
&y_c(t) = C_c(\rho_c(t))  x_c(t) + D_c(\rho_c(t))  u_c(t)
\end{align}
%\vspace{-1.6em}

Let $\mathcal{A}_p$ and $\mathcal{A}_c$ denote the admissible parameter trajectories of $G_1$ and $G_2$ respectively. Let $G_2$ be conic with sector bounds $[a_c,b_c]$ and $G_1$ be intermittently conic with sector bounds $[a_p,b_p]$. We also assume $||u_p(t)||_2^2 \in [\bar u_{p,1}^2, \bar u_{p,2}^2]$ with constants $\bar u_{p,1}>0$ and $\bar u_{p,2}<\infty$.  We now examine the stability of the negative feedback interconnection of $G_1$ and $G_2$ when $G_1$ is intermittently conic as defined in Section \ref{problem}.
\begin{figure}[!t]
	\centering
	\includegraphics[scale=0.27]{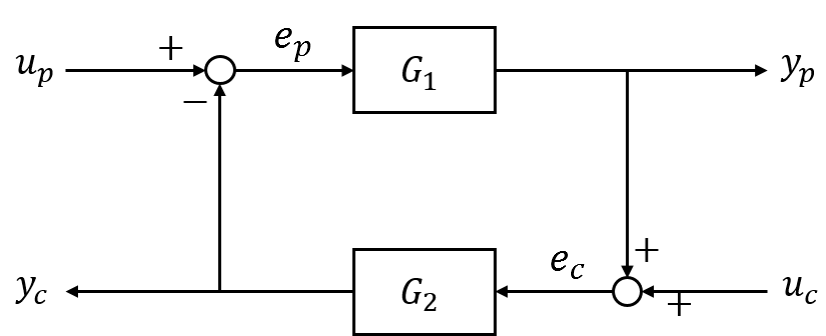}
	\caption{Negative feedback interconnection of $G_1$ and $G_2$}
	\vspace{-1.5em}
	\label{negativefeedback}
\end{figure}
\vspace{-0.8em}
\begin{theorem}[Modified Conic Sector Theorem]\label{oneconic}
 The negative feedback interconnection of $G_1$ and $G_2$ as shown in Fig. \ref{negativefeedback} is $\mathcal{L}_2$ stable if, for every admissible $\rho_p(t) \in \mathcal{A}_p$ and $\rho_c(t) \in \mathcal{A}_c$, $t_n \to \infty$,
	%	\begin{itemize}
	%		\item[(i)] 
	\vspace{-0.5em}
	\begin{equation}\label{oneconic1}
	\int \limits_{t \in {t_{c}}} {\epsilon_p}\left( {\rho_p \left( t \right)} \right)dt \geq \bar{u}_{p,1}^2 \int \limits_{t \in {t_{nc}}} \alpha_p \left( {\rho_p \left( t \right)} \right)dt, \quad \text{and,}
	\end{equation} 
	%		\item[(ii)] 
	\vspace{-0.5em}
	{
	\begin{equation}\begin{gathered}
		\frac{1}{b_c}{+}a_p>0, \; \left(\frac{1}{b_c}{+}a_p\right)\left(\frac{1}{b_p}{+}a_c\right){-}\frac{1}{4}\left(\frac{a_c}{b_c}{-}\frac{a_p}{b_p}\right)^2>0.
		\label{oneconic2}
		\end{gathered}\end{equation} }
\end{theorem}
\begin{proof}
Since $G_2$ is conic, for all $t_1$ and~$t_2$ such that~$t_0\leq t_1 <t_2\leq t_n$, it satisfies \setlength{\arraycolsep}{-0.8pt}
{\small
	\vspace{-0.2em}
\begin{equation}
\label{g1}
%\frac{1}{t_N-t_0}
\mathop \int_{t_1}^{t_2}\left[\begin{array}{c} y_c(t)\\u_c(t)\end{array}\right]' 
\left[ \begin{array}{lr}-\frac{1}{b_c}I&\frac{1}{2}\left(1+\frac{a_c}{b_c}\right)I \\\frac{1}{2}\left(1+\frac{a_c}{b_c}\right)I&-a_cI\end{array}\right]
\left[\begin{array}{c}y_c(t)\\u_c(t)\end{array}\right]dt \geq 0.
\vspace{-0.2em}
\end{equation}}If \eqref{oneconic1} holds, then $G_1$ is conic in the average sense, with \setlength{\arraycolsep}{-1pt}
{\small
	\vspace{-0.2em}
	\begin{equation}
	\label{g2}
	%\frac{1}{t_N-t_0}
	\mathop \int_{t_1}^{t_2}\left[\begin{array}{c} y_p(t)\\u_p(t)\end{array}\right]' 
	\left[ \begin{array}{lr}-\frac{1}{b_p}I&\frac{1}{2}\left(1+\frac{a_p}{b_p}\right)I \\\frac{1}{2}\left(1+\frac{a_p}{b_p}\right)I&-a_pI\end{array}\right]
	\left[\begin{array}{c}y_p(t)\\u_p(t)\end{array}\right]dt \geq 0.
	\vspace{-0.2em}
	\end{equation}}From Fig. \ref{negativefeedback}, we have $e_c=u_c+y_p$ and $e_p=u_p-y_c$, which when used in conjunction with \eqref{g1} and \eqref{g2} gives
		\setlength{\arraycolsep}{1pt}
		\vspace{-0.4em}
		\begin{equation}
		\label{l2_intermediate}
		%\frac{1}{t_N-t_0}
		\mathop \int_{t_0}^{t_n} \left[\begin{array}{c} Y(t) \\ E(t)\end{array}\right]' 
		\left[ \begin{array}{lr}Q&S \\S'&R\end{array}\right]
		\left[\begin{array}{c}Y(t)\\E(t)\end{array}\right]dt \geq 0,
		\vspace{-0.2em}
		\end{equation}
where $Y{=}[y_c^{'},y_p^{'}]$, $E{=}[e_c^{'},e_p^{'}]$, and	\setlength{\arraycolsep}{-0.5pt}
{\small
\begin{equation}\label{qsr}
\begin{gathered}
\vspace{-0.2em}
Q{=}\left[{\begin{array}{*{20}{c}}
	-\frac{1}{b_c}I{-}a_pI& \frac{a_p}{b_p}{-}\frac{a_c}{b_c}\\
	*&-\frac{1}{b_p}I{-}a_cI
	\end{array}}\right], R{=}\left[{\begin{array}{*{20}{c}}
	-a_c I&0\\
	0&-a_p I
	\end{array}}\right],\, \text{and} \\ S{=}\left[{\begin{array}{*{20}{c}}
	\frac{1}{2}\left(1{+}\frac{a_c}{b_c}\right)I&-a_p I\\
	-a_c I&\frac{1}{2}\left(1{+}\frac{a_p}{b_p}\right)I
	\end{array}}\right].\\
\end{gathered}
\vspace{-0.2em}
\end{equation}}Setting $t_0=0$ and $t_n \to \infty$ in (\ref{l2_intermediate}), the feedback interconnection is $\mathcal{L}_2$ stable if $Q<0$, giving condition (\ref{oneconic2}).
\end{proof}
\begin{remark}
The modified conic sector theorem (Theorem \ref{oneconic}) provides sufficient conditions for stability of the closed-loop system based on the plant and feedback controller parameters. This conic sector theorem can be graphically interpreted as follows. The feedback interconnection of the plant and the controller is stable if the two systems are conic in the average sense (Equation \ref{eq:av_conicity}) and the conic sectors they lie in do not intersect. {Note that Theorem \ref{oneconic} is an input-output stability result; however, with additional assumptions like zero state detectability and reachability, it can be extended to prove Lyapunov stability along the lines of \cite{hill1976stability}.}
\end{remark}
%Note that the above results are symmetric, i.e., they also hold when $G_2$ is intermittently conic and $G_1$ is strictly conic.

\section{Conic Sector Theorem Based Control Design}\label{controldesign}
In this section, we present a control design approach based on the modified conic sector theorem derived in Section \ref{conicsector}. 
\vspace{-1.2em}
\begin{figure}[b]
	\centering
	\vspace{-1.2em}
	\includegraphics[scale=0.43]{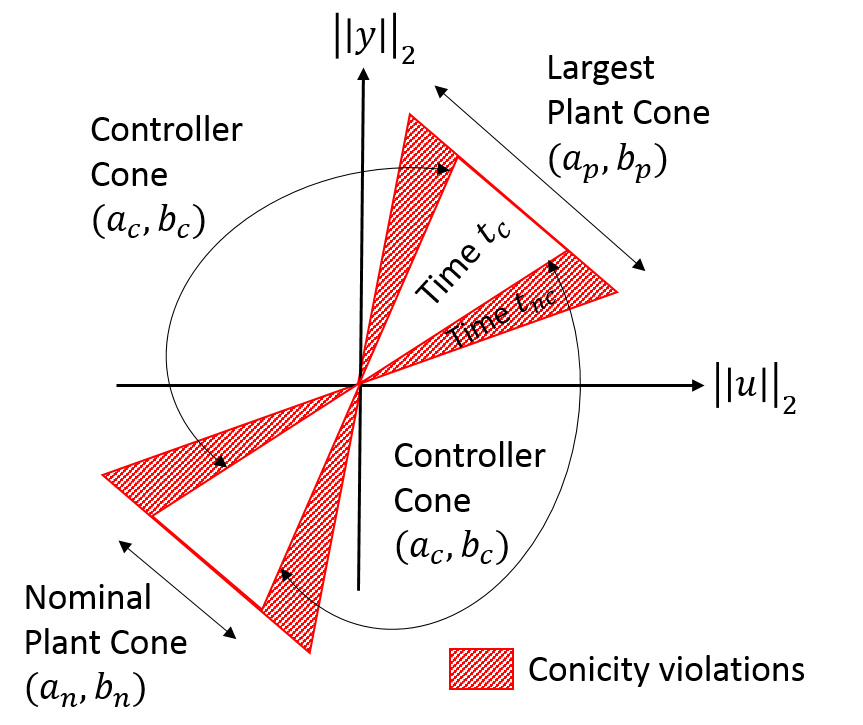}
	\caption{Conic-sector-based control design - key idea}
	\label{idea}
\end{figure}
\begin{figure*}
	\begin{minipage}[c]{\textwidth}
		% ensure that we have normalsize text
		\scriptsize
		% Store the current equation number.
		\newcounter{MYtempeqncnt}
		\setcounter{MYtempeqncnt}{\value{equation}}
		% Set the equation number to one less than the one
		% desired for the first equation here.
		% The value here will have to changed if equations
		% are added or removed prior to the place these
		% equations are referenced in the main text.
		% The spacer can be tweaked to stop underfull vboxes.
		%	\vspace*{1pt}
		\setcounter{equation}{17}
		\renewcommand{\arraystretch}{1}
		\setlength{\arraycolsep}{1.2pt}
		\begin{equation}\label{coniccontrol}
		\left[ {\begin{array}{*{20}{c}}
			{A_{cl}^{'} P(\rho_p) + P(\rho_p){A_{cl}}+\dot \rho_p \frac{\partial{P(\rho_p)}}{\partial \rho_p} + {C_{cl}^{'}}Q C_{cl}}&{P(\rho_p)B_{cl} - C_{cl}^T   S + {C_{cl}^{'}}Q D_{cl}} \\ 
			{{B_{cl}^{'}}P(\rho_p) - {  S^{'}}C_{cl} + {D_{cl}^{'}}QC_{cl}}&{  {D_{cl}^{'}}QD_{cl} - {  S^{'}}D_{cl} - {D_{cl}^{'}}S_{cl} +   R}
			\end{array}} \right]+\epsilon_{cl} I \leq 0
		\end{equation}
		\begin{equation}\label{nonconiccontrol}
		\left[ {\begin{array}{*{20}{c}}
			{A_{cl}^{'} \tilde P(\rho_p) + \tilde P(\rho_p){A_{cl}}+\dot \rho_p \frac{\partial{\tilde P(\rho_p)}}{\partial \rho_p} + {C_{cl}^{'}}  Q C_{cl}}&{\tilde P(\rho_p)B-C_{cl}^{'}   S + {C_{cl}^{'}}  QD_{cl}} \\ 
			{{B_{cl}^{'}}\tilde P(\rho_p) - {  S^{'}}C_{cl} + {D_{cl}^{'}}  QC_{cl}}&{{D_{cl}^{'}}QD_{cl} - {  S^{'}}D_{cl} - {D_{cl}^{'}}S +   R-\alpha_{cl} I}
			\end{array}} \right] \leq 0
		\end{equation}
		\begin{equation}\label{condition1}
		\int_{t_c} \epsilon_{cl}=\int_{t_{nc}}\alpha_{cl}, \quad P(\rho_p)\geq 0, \quad \tilde P(\rho_p)\geq 0
		\end{equation}
		% Restore the current equation number.
			\setcounter{equation}{\value{MYtempeqncnt}}
		% IEEE uses as a separator
		\hrulefill
		\end{minipage}
		\vspace{-1.5em}
\end{figure*}

\textbf{Key design idea:} Consider an LPV plant that lies in a range of conic sectors as depicted in Fig. \ref{idea}. A conservative control design approach involves designing the controller such that it lies in a cone that is complementary to the largest conic sector that contains the plant dynamics \cite{joshi2002design}. However, this strategy may result in very low gain controllers. Therefore, we propose the following alternative design. We choose a nominal plant cone that is smaller than the largest plant cone, and design a controller that lies in the conic sector complementary to this nominal cone, based on Theorem \ref{oneconic}. Now, the controller cone will intersect the plant cone in certain parameter regions, making the closed-loop system nonconic in these regions. Based on Theorem \ref{res2}, we then bound the fraction of time that the closed-loop system operates in these nonconic regions. Note that in contrast to standard LPV design methods, we do not use the parameter trajectory for gain scheduling. Rather, we use the knowledge of the parameter trajectory to characterize conic and nonconic parameter regions, which are then used to design a single conic controller for a nominal operating point.

\textbf{Design Procedure:}
Let $\gamma$ be the induced $\mathcal{L}_2$ gain of the feedback interconnection of plant $G_1$ and controller $G_2$ shown in Fig. \ref{negativefeedback}. We would like to choose a nominal plant cone $[a_p^{*},b_p^{*}]$ such that a controller $G_2$ lying in the conic sector $[a_c^{*}=-\frac{1}{b_p^{*}}+\delta,b_c^{*}=-\frac{1}{a_p^{*}}-\delta]$, where $\delta >0$ is a trivially small number, minimizes $\gamma$ while keeping the closed-loop system $\mathcal{L}_2$ stable. This design problem can be formulated as
\vspace{-1.1em}
\begin{equation}
P_1: \quad  \min_{a_p,b_p} \gamma \quad \;\text{s.t.}\int_{t_c}\epsilon_{cl} dt \geq \int_{t_{nc}}\alpha_{cl} dt,
\vspace{-0.8em}
\label{optimal}
\end{equation}
where $\epsilon_{cl}$ and $\alpha_{cl}$ are the conicity indices as defined in (\ref{conic_eqn}) and (\ref{nonconic_eqn}) respectively for the closed-loop system. We now have the following result.
\begin{prop}\label{res_l3}
	If there exist $a_p=a_p^*$ and $b_p=b_p^*$ satisfying matrix inequalities (\ref{coniccontrol})-(\ref{condition1}), where $A_{cl}, B_{cl}, C_{cl}$ and $D_{cl}$ are the closed loop system matrices, and $Q$, $S$ and $R$ are defined in \eqref{qsr}, then the conic controller $G_2$ with $[a_c,b_c]=[-\frac{1}{b_p^{*}},-\frac{1}{a_p^{*}}]$ renders the closed-loop system in Fig.~\ref{negativefeedback} $\mathcal{L}_2$ stable. Further, $a_p^*$ and $b_p^*$ solve problem $P_1$.
\end{prop}
\begin{proof}
Let $\beta>0$ be the smallest scalar such that $\beta I+R>0$ and $\left(1/{\beta}+Q\right)<\zeta I$, for some $\zeta>0$. From \eqref{l2_intermediate}, with $t_0=0$ and $t_n\to \infty$, we have
	\addtocounter{equation}{3}
%	\begin{equation}
$	\int_{0}^{\infty} Y'(t)Y(t)dt \geq -\frac{(\beta +\lambda(R))}{\zeta}  \int_{0}^{\infty} U'(t)U(t)dt$,
%	\end{equation}
	where $Y=[y_c^{'},y_p^{'}]$, $U~=~[u_c^{'},u_p^{'}]$ and $\lambda(R)$ is the maximum eigenvalue of $ R$. Substituting $a_c=-\frac{1}{b_p}$ and $b_c=-\frac{1}{a_p}$, the closed-loop $\mathcal{L}_2$ gain can be written in terms of the plant parameters as 
%	\begin{equation}\label{gammaqsr}
	$\gamma=(\beta +\lambda(R))/{\zeta}$.
%	\end{equation}
	The numerator and denominator of $\gamma$ can be expressed as $k_1 \left({b_p}/{a_p}-1\right)$ and $k_2\left({b_p}/{a_p}\right)$ respectively, where $k_1$ and $k_2$ are constants independent of the optimization variables $a_p$ and $b_p$. Therefore, the optimization problem $P_1$ is equivalent to
	\vspace{-0.5em}
	\begin{equation}
	P_2: \, \min_{a_p,b_p} r_p=\frac{(b_p-a_p)}{2}, \, \text{s.t.}\int_{t_c}\epsilon_{cl}dt \geq \int_{t_{nc}}\alpha_{cl}dt,
	\vspace{-0.5em}
	\end{equation}
	that is, minimizing the conic radius of the nominal plant cone. %Define $M_{cl}=\int_{t_c}\epsilon_{cl} - \int_{t_{nc}}\alpha_{cl}$ to be the conicity margin of the closed-loop system. 
	As $r_p$ is increased, the intersection of the plant and controller cones increases. Therefore, the solution to problem $P_2$ is
%	\begin{align}\label{mcl}
$	[a_p^{*},b_p^{*}]=\{[a_p,b_p]:\int_{t_c}\epsilon_{cl}dt - \int_{t_{nc}}\alpha_{cl}dt=0\}$,
%	\end{align}
	which is the smallest plant cone such that the closed-loop system remains conic in the average sense defined in Theorem \ref{res2}, implying (\ref{condition1}). 
	Equations (\ref{coniccontrol}) and (\ref{nonconiccontrol}) correspond to computation of $\epsilon_{cl}$ and $\alpha_{cl}$ (see Appendix \ref{determining}). 
\end{proof}
\section{Application Example: Power System with Very High Penetration of Renewable Energy}\label{simulation}
A class of systems where the proposed control design finds application is power grids with a large penetration of renewable energy resources. The dynamics of power grids are described by nonlinear differential algebraic equations, which are extremely difficult to analyze. Therefore, LPV models of power grids have been proposed for the design of wide-area controllers \cite{qiu2004decentralized}. LPV models and switched system approximations have also become relevant in the context of time-varying generation from renewable energy sources \cite{agarwal2017feedback}. Typical approaches to stabilize power grids with LPV models involve the design of parameter varying controllers \cite{qiu2004decentralized} that can be difficult to design and implement. In this section, we demonstrate the design of a conic controller to stabilize power grids with very high penetration of renewable energy, enhance their peak loading capacity and minimize losses.

{ Let the generation sources $G_1,G_2,...,G_{N_g}$ in the grid be represented by LPV systems of the form \eqref{lpvsys1} with the parameters of the $i$-th generator being $\rho_{p,i}=(P_i,Q_i)$, representing the scheduled real and reactive power, and the output being $y_i{=}[v_i]$, representing the voltage at the $i$-th generator bus.} The detailed nonlinear models of the renewable generators and the electrical network can be found in \cite{sivaranjani2013networked}. The system power flow constraint requires
\vspace{-0.7em}
\begin{equation}\label{powerflow}
P_L=\sum_{i=1}^{N_g} {P_i}-P_{loss}\;,\; Q_L=\sum_{i=1}^{N_g} {Q_i}-Q_{loss},
\vspace{-0.5em}
\end{equation}
where the total system load is $P_L+jQ_L$ and the power losses are 
%\begin{align*}
$P_{loss}=0.5\left(\sum_{j=1}^{N} \sum_{i=1}^N |{i_{ij}|^2 r_{ij}}\right)$ and $Q_{loss}=0.5\left(\sum_{j=1}^{N} \sum_{i=1}^N |{i_{ij}|^2 x_{ij}}\right)$, with
%\end{align*}
$i_{ij}$, $r_{ij}$ and $x_{ij}$ being the line currents, resistances and reactances respectively. Our objective is to minimize the total power loss in the system,
%\begin{equation}\label{globalopt}
%\begin{gathered}
%\mathop {\min }\limits_{P_i,Q_i} ({P_{loss}+Q_{loss}})\;,\quad\text{s.t.}\,(\ref{powerflow}).
%\end{gathered}
%\end{equation}
%We now need to design a conic controller for every generator such that this objective is achieved. This design problem (\ref{globalopt}) 
which is equivalent to minimizing the induced $\mathcal{L}_2$ norm from the output vector $y_i$ to the input vector $u_i$, subject to \eqref{powerflow}. 

We consider a modified version of the IEEE 14-bus standard test system as shown in Fig. 3 to illustrate the proposed design. In this modified system, the generator at bus 1 is replaced with an equivalent 600 MW doubly-fed induction generator (DFIG) wind farm (78\% renewable energy penetration). {The system is then linearized around the power flow solution at 4 operating points of $(P_{DFIG},Q_{DFIG})$ (Table I) and a state space model with 48 states, 5 inputs and 5 outputs is obtained. The five control inputs correspond to the reference signals of the Automatic Voltage Regulator (AVR) excitation controllers on the generators at buses 2, 3, 6 and 8, and the speed reference signal of the pitch angle controller on the DFIG at bus 1.} The conic sector bounds for each operating point are listed in Table I. %The maximum loading capacity of the DFIG is found to be approximately 460 MW.  
\begin{figure}[t]
 		\centering
 		\vspace{-0.45em}
	\includegraphics[trim=3cm 1cm 3cm 0cm,scale=0.205]{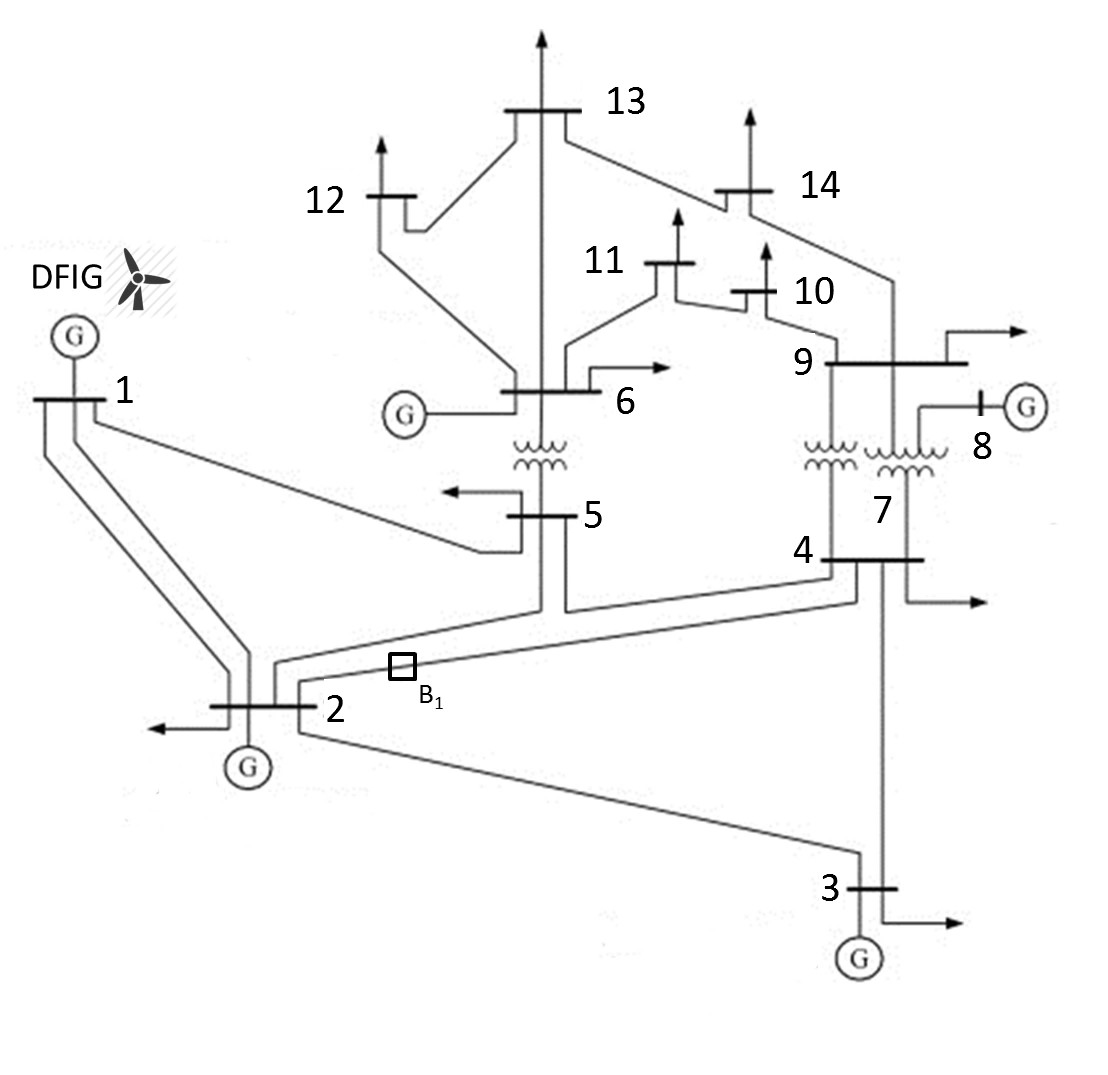}
	\vspace{-1.2em}
\captionof{figure}{Modified IEEE 14-bus test system.}
\vspace{-1em}
\label{14Busmod1}
\end{figure}
	\begin{figure}
	\centering
	\vspace{-0.35em}
	\includegraphics[trim=1cm 0cm 1cm 0.15cm,scale=0.22]{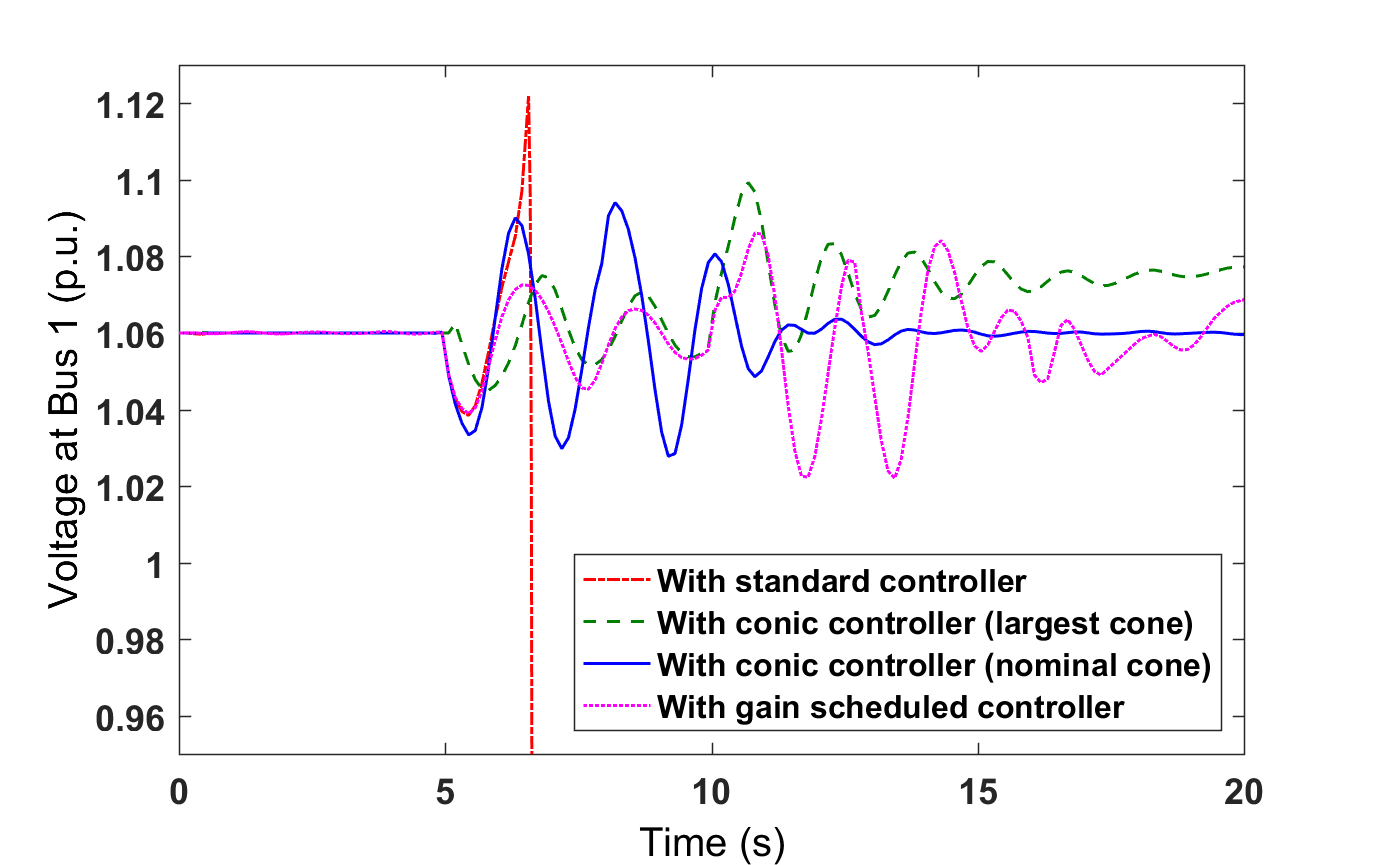}
	\vspace{-0.35em}
	\captionof{figure}{Voltage profiles with conic and standard controllers.}\label{voltage}
	\vspace{-2em}
\end{figure}
We solve (\ref{coniccontrol})-(\ref{condition1}) over a grid of parameter values using YALMIP/SeDuMi \cite{lofberg2004yalmip,sturm1999using} to obtain the nominal operating point of $(P_{DFIG},Q_{DFIG})$ = (350 MW,-28.2 MVAr) and the corresponding nominal plant and controller conic bounds. { We then synthesize a conic controller with system matrices $(A_c,B_c,C_c)$ of dimensions $(48{\times}48)$, $(48{\times}5)$ and $(5{\times}48)$ respectively, satisfying these nominal conic bounds using the procedure in \cite[Section 5]{bridgeman2014conic}.} We also implement three other controllers for comparison, namely, the classical AVR and pitch angle controllers, referred to as the `standard controller', as described in \cite{milano2005open}, a conic controller designed to lie in a sector complementary to the largest plant cone (corresponding to the operating point $P_{DFIG}=$ 420 MW) \cite{joshi2002design}, and a dynamic gain scheduled controller designed to stabilize the system for every operating point \cite{qiu2004decentralized}.
	\begin{table}
	\centering
	\vspace{0em}
	\captionof{table}{Operating points, conic sector bounds and average power losses.}
	\label{simulationtable}
	\footnotesize
	\begin{tabular}{lllllll}
		\vspace{-0.1em}
		$P_{DFIG}$                                         &  $Q _{DFIG}$                            & $a_p$                                  & $b_p$                                 & $P_{loss}$                            & $Q_{loss}$    &   $S_{loss}$                   \\
		(MW)                                        & MVAr &  & & (MW)   & (MVAr)                           & (MVA)                          \\
		%	500                                           & -35                                   & \multicolumn{2}{l}{Nonconic}                                                  & 57.05                                 & 198.41                                \\
		420                                           & -34.00                                   & -3.499 &	7.501
		& 46.68                                 & 157.90  & 164.65                               \\
		400                                           & -31.70                                 & -3.199 &	7.201                                & 37.34                                 & 121.32     &126.94                           \\
		{\color{red} \textbf{350}} & {\color{red} \textbf{-28.20}} & {\color{red} \textbf{-2.425}} & {\color{red} \textbf{7.575}} & {\color{red} \textbf{29.43}} & {\color{red} \textbf{90.27}} & {\color{red} \textbf{94.95}} \\
		{\color{red} \textbf{(Nominal)}} & & & & & & \\
		300                                           & -23.05                                & -2.899 &	6.901
		& 21.93                                 & 60.72      & 64.55        
	\end{tabular}
	\vspace{-3.15em}
\end{table}
{ 
		
	We test the performance of the closed-loop system when all real and reactive loads are increased by 15$\%$ between $t=5s$ to $t=10s$ (Fig. \ref{voltage}). This test condition is chosen to compare the responses of the controllers to two major changes in the system parameter (operating point) corresponding to the increase in the total load at $t=5s$ and the decrease in the system load at $t=10s$.} With the standard controller for the 14-bus system \cite{milano2005open}, the system loses stability for the  15$\%$ load change, implying that the maximum loading capability of the system is exceeded. With the conic controller designed for the largest plant cone \cite{joshi2002design}, the system remains stable but exhibits poor performance, with the voltage at the DFIG bus not settling at the nominal value. {With the classical gain scheduled controller, the system displays adequate transient response to the first operating point change at $t=5s$, but responds poorly to the second change at $t=10s$. On the other hand, the system with the conic controller designed for the nominal cone responds well to this load change and maintains a voltage level close to the nominal value, which is an advantage of the increased controller gain obtained by allowing the intermittently conic operation. This result indicates that the conic controller enhances the maximum loading capability of the grid. We also observe from Table I that the nominal conic controller has a significantly lower net power loss (94.95 MVA) as compared to the conic controller designed for the largest plant cone (164.65 MVA) and the classical gain scheduled controller (116.36 MVA).}
\vspace{-0.46em}
\section{Conclusion}
We derived a modified conic sector theorem for intermittently conic LPV systems, and employed this theorem to synthesize conic controllers that are less conservative than traditional designs. We also demonstrated the performance and advantages of this design in stabilizing a power grid with very high penetration of renewable energy.
\vspace{-0.46em}
\appendices
\section{Determining Conic Bounds and Conicity Indices}\label{determining}
\begin{theorem}
	\label{conicbounds}
	The LPV system (\ref{lpvsys}) is conic with sector bounds $[a,b]$ if and only if there exists a symmetric $\hat P(\rho)>0$ such that for all $\rho \in \mathcal{A}$, the following holds:
	{\footnotesize
		\vspace{-0.5em}
		\begin{equation}\label{coniclmi}
		\left[ {\begin{array}{*{20}{c}}
			{A' \hat P(\rho) + \hat P(\rho){A}+\dot \rho \frac{\partial{\hat P(\rho)}}{\partial t} +{C'} C}&{\hat P(\rho)B - \frac{a+b}{2}C'  + {C'}D} \\ 
			*&{  {D'}D - \frac{a+b}{2}(D + {D'}) + abI}
			\end{array}} \right]{\leq}0.
		\vspace{-0.3em}
		\end{equation}}
\end{theorem} 
The proof is along the lines of \cite[Theorem 2.1]{joshi2002design} and is omitted here. 
In conic parameter regions, $\exists P(\rho){>}0 $ such that 
\vspace{-0.8em}
{\small
\begin{equation}\label{lyapunov1}
{A}' P(\rho) +  P(\rho){A} < 0.
\vspace{-0.4em}
\end{equation}}From (\ref{conic_eqn}), following standard literature \cite{wu1996}, and defining $Q=-1/b$, $S=\frac{1}{2}(1+\frac{a}{b})$, and $R=-a I$, we can derive% the inequality 
{\scriptsize
	\vspace{-0.4em}
	\setlength{\arraycolsep}{-0.9pt}
\begin{equation}\label{conicindex}
\left[ {\begin{array}{*{20}{c}}
	{A^{'} P(\rho) +  P(\rho){A}+\dot \rho \frac{\partial{ P(\rho)}}{\partial t} + {C'} Q C}&{ P(\rho)B - C'  S + {C'} Q D} \\ 
	{{B'} P(\rho)-{ S'}C + {D'} Q C}&{ {D'}QD-{ S'}D-{D'}S + R}
	\end{array}} \right]+\epsilon I{\leq}0.
\end{equation}}This matrix inequality can then be solved with constraints (\ref{rhobound}) and (\ref{lyapunov1}) to determine $\epsilon(\rho(t))$. Similarly, from (\ref{nonconic_eqn}), we can obtain the matrix inequality 
{\scriptsize
	\setlength{\arraycolsep}{-2.7pt}
	\vspace{-0.5em}
\begin{align}\label{nonconicindex}
\left[ {\begin{array}{*{20}{c}}
	{A' \tilde P(\rho) + \tilde P(\rho){A}+\dot \rho \frac{\partial{\tilde P(\rho)}}{\partial t} + {C'} Q C}&{\tilde P(\rho)B-C'  S + {C'} Q D} \\ 
	{{B'}\tilde P(\rho) - { S'}C + {D'} Q C}&{{D'}Q D - {S'}D - {D'}S + R-\alpha I}
	\end{array}} \right]{\leq}0, 
\end{align}}where $\tilde P(\rho)>0$, which can be solved with bounds (\ref{rhobound}) to determine $\alpha(\rho(t))$.
The matrix inequalities to compute the conic bounds and conicity indices are infinite dimensional due to their dependence on $\rho$ and can be solved by discretization over a mesh on $\rho$ and $\dot \rho$. Further, the conic bounds obtained are not unique; the computation of tight conic bounds for each parameter is discussed in \cite{bridgeman2014conic}.
\vspace{-0.5em}
\section{} \label{proofs}
\textit{Proof of Lemma \ref{res1}:}
	We have
	\setlength{\arraycolsep}{0pt}
	\vspace{-0.8em}
	{\small
		\begin{equation}\label{I}
		\mathop \int\limits_{t_0}^{t_n}{w(u,y)dt}=\sum\limits_{i\in i_{c}} \left( \int\limits_{{t_{i-1}}}^{t_i}w(u,y)dt  \right)+\sum \limits_{i\in i_{nc}} \left(\int\limits_{t_{i-1}}^{t_{i}} w(u,y)dt  \right).
		\end{equation}\vspace{-0.8em}}From \eqref{I}, (\ref{conic_eqn}), (\ref{nonconic_eqn}) and Assumption \ref{input}, we have
	\setlength{\arraycolsep}{0pt}
	{\small
		\begin{align}\label{bound}
		\mathop \int_{t_0}^{t_n}	w(u,y)dt &\geq  \sum\limits_{i\in i_{c}} \left(\epsilon_{i}(t_i-t_{i-1})\right)  \nonumber\\&-\sum\limits_{i\in i_{nc}}\left(\alpha_{i}\bar{u}_1^2(t_i-t_{i-1})\right).\vspace{-1.2em}
		\end{align}
	}Comparing (\ref{bound}) with \eqref{eq:av_conicity}, system (\ref{lpvsys}) is conic in the average sense if, for every $\rho(t) \in \mathcal{A}_d$, (\ref{res_1}) is satisfied.
	%	\setlength{\arraycolsep}{0pt}
	%%	\vspace{-1.5em}
	%	{%\small
	%		\begin{align}\label{p2}
	%		& \sum\limits_{i\in i_{c}} \left(\sup \limits_{i}\{\epsilon_{i}(\rho(i))\}(t_i-t_{i-1})\right) \nonumber\\ &\geq \sum\limits_{i\in i_{nc}}\left(\sup \limits_{i}\{\alpha_{i}(\rho(i))\}(t_i-t_{i-1})\right)
	%		\end{align}
	%	}is satisfied, which is equivalent to (\ref{res_1}).
%\end{proof}

\textit{Proof of Theorem \ref{res2}:}
	Consider the discrete parameter trajectory $\rho(i)$ in Lemma \ref{res1}. Then, the interval $\mathcal{I}=~[t_0,t_n]=\bigcup_{i=1}^{n}[t_{i-1},t_i]=\left(\bigcup_{i \in i_{c}}[t_{i-1},t_i]\right) \bigcup \left(\bigcup_{i \in i_{nc}}[t_{i-1},t_i]\right)$.
	%$. Then, we have
	%	\begin{equation*}
	%	\mathcal{I}=\bigcup_{i=1}^{N}[t_{i-1},t_i]=\left(\bigcup_{i \in i_{c}}[t_{i-1},t_i]\right) \bigcup \left(\bigcup_{i \in i_{nc}}[t_{i-1},t_i]\right).
	%	\end{equation*}
	Define a partition of $\mathcal{I}$ to be $P=\{t_0,t_1,\ldots,t_n\}$ with norm $|P|=\max{(t_i-t_{i-1})}$, $i \in\mathbb{N}_n$. Let $\mathcal{R}=~\{\rho(1),\rho(2),\ldots,\rho(n)\}$ such that for every $i \in\mathbb{N}_n$, $\rho(i) \in [t_{i-1},t_i]$. Define the tagged partition $\mathcal{F}(t,\rho)=(P,\mathcal{R})$ to be the collection of the partition $P$ and the associated parameter set $\mathcal{R}$. Define the Riemann sum $S_r$ with respect to tagged partition $\mathcal{F}(t,\rho)$ to be the RHS of \eqref{bound}.
	In the limit $|P| \to 0$, since $\rho(t)$ is piecewise continuously differentiable, $S_{r}$ can be approximated by the Riemann integral
	\vspace{-0.7em}
	\begin{equation}\label{sums}
	\lim\limits_{|P| \to 0}S_r=\int_{t \in t_c}  \epsilon(\rho(t)) dt-\bar{u}_1^2\int_{t \in t_{nc}}\alpha (\rho(t)) dt.
%	\vspace{-0.7em}	
	\end{equation}
	From (\ref{bound}) and (\ref{sums}), \eqref{lpvsys} is conic in the average sense if \eqref{res_p} is satisfied for every $\rho(t) \in \mathcal{A}$, completing the proof.
	\vspace{-0.3em}
\bibliographystyle{IEEEtran}
\bibliography{LPVandIQCfurther}

%%%%%%%%%%%%%%%%%%%%%%%%%%%%%%%%%%%%%%%%%%%%%%%%%%%%%%%%%%%%%%%%%%%%%%%%%%%%%%%

\end{document}